\documentclass[%
 aip,
 amsmath,amssymb,
reprint, onecolumn 
]{revtex4-1}

\usepackage{graphicx}
\usepackage{dcolumn}
\usepackage{bm,mathrsfs,bbold,upgreek}

\usepackage[utf8]{inputenc}
\usepackage[T1]{fontenc}
\usepackage{mathptmx}
\usepackage{etoolbox}

\makeatletter
\def\@email#1#2{%
 \endgroup
 \patchcmd{\titleblock@produce}
  {\frontmatter@RRAPformat}
  {\frontmatter@RRAPformat{\produce@RRAP{*#1\href{mailto:#2}{#2}}}\frontmatter@RRAPformat}
  {}{}
}%
\makeatother

\newtheorem{theorem}{Theorem}

\newtheorem{corollary}[theorem]{Corollary}

\newtheorem{definition}[theorem]{Definition}

\newtheorem{proposition}[theorem]{Proposition}
\newtheorem{remark}[theorem]{Remark}

\newenvironment{proof}[1][Proof]{\noindent\textbf{#1.} }{\ \rule{0.5em}{0.5em}}
\begin{document}


\title{An Approach to Anti-Wick Ordering of Bosonic Fields} 



\author{John Gough}
\affiliation{Aberystwyth University, SY23 3BZ, Wales, United Kingdom}

\author{H. Yamashita}

\affiliation{Aichi-Gakuin University
Nissin, Aichi,
Japan}

\email{jug@aber.ac.uk}
\email{yamasita@dpc.aichi-gakuin.ac.jp}
\date{\today}

\begin{abstract}
We present a new technique for putting general boson fields into ant-Wick ordered form. The anti-Wick map associates an operator with a given function of complex variables, and we show that it may be realized as composition of a mapping to a commutative sub-algebra of a doubled-up boson algebra followed by a partial conditional expectation onto one of the factors.
\end{abstract}


\maketitle 

\section{Introduction}
Let $\hat A$ be a Bose annihilation operator so that $[\hat A , \hat A^\ast ] = \hat I$. The anti-Wick ordering (also known as anti-normal, or Berezin, or Toeplitz ordering) of a function $f(\alpha^\ast , \alpha ) = \sum_{nm} f_{nm} (\alpha^\ast)^n \alpha^m $ of a complex variable $\alpha$ is defined as $\mathscr{A} (f) = \sum_{nm} f_{nm} \hat A^m (\hat A ^\ast)^n$. Anti-normal ordering is widely used in quantum field theory \cite{Berezin} and quantum optics \cite{Louisell}, and is essential for the Sudarshan-Glauber P-representation \cite{Sudarshan,Glauber}.

In this paper, we wish to elaborate on a remarkable procedure for realizing the anti-Wick ordered quantization which was recently discovered by one of the authors \cite{Yamashita22}. 
The starting point is the observation that if we introduce a second Bose annihilation operator $\hat B$ so that $ [\hat B , \hat B^\ast ] = \hat I$ with all other commutators between $\hat A , \hat A^\ast , \hat B , \hat B^\ast$ vanishing. We observe that $\hat C = \hat A + \hat B^\ast$ is a normal operator: $[ \hat C ,\hat C^\ast]= 0 $. 
From this, we may introduce the mapping $\varphi :f(\alpha ^\ast , \alpha )\mapsto f(\hat  C^\ast , \hat C)$ which we note is a mapping between \textit{commutative} algebras. This followed by taking the partial vacuum expectation $\mathcal{E}$ with respect to the $\hat B$'s. This results in the anti-Wick quantization:
\begin{eqnarray}
    \mathscr{A} \equiv \mathcal{E} \circ \varphi .
\end{eqnarray}

To see this, let us take $\mathcal{H}_A$ and $\mathcal{H}_B$ to be the Hilbert spaces for the two modes so that $\hat C = \hat A \otimes \hat I + \hat I \otimes \hat B^\ast$ on $\mathcal{H}_A \otimes \mathcal{H}_B$ with $[\hat C , \hat C^\ast ] = \hat I$. The partial vacuum expectation $\mathcal{E}$ taking us from operators on $\mathcal{H}_A \otimes \mathcal{H}_B$ to operators on $\mathcal{H}_A $ is extended by linearly from the basic feature $: T_A \otimes T_B \to \langle \Omega |T_B \,\Omega \, \rangle \, T_A$ where $\Omega$ is the vacuum (i.e., $\hat B \, \Omega =0$). Then, $\mathcal{E} f (\hat C^\ast , \hat C^\ast ) = \mathscr{A} (f)$. To see, let us set $f(\alpha^\ast , \alpha )= (\alpha ^\ast )^n \alpha^m$ then
\begin{eqnarray}
\mathcal{E} :
(\hat A^\ast \otimes \hat I + \hat I \otimes \hat B^\ast)^n
(\hat A \otimes \hat I + \hat I \otimes \hat B)^m
     \to \hat A^m \, (\hat A^\ast )^n \equiv \mathscr{A} (f).
\end{eqnarray}
The construction first replaces the classical variables with commuting operators (which may be reordered without restriction!) only to have the partial vacuum expectation kill off everything except the Wick-ordered terms in $\hat B, \hat B^\ast$: the surviving terms then being the anti-Wick ordered terms in $\hat A, \hat A^\ast$.
This somewhat surprising property stems from the fact that the map $\mathcal{E}$ actually takes us from a commutative algebra into a non-commutative one.

\begin{figure}[h!]
    \centering
    \includegraphics[width=0.40\linewidth]{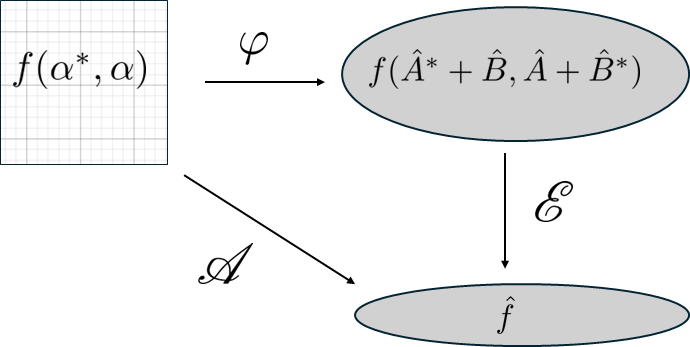}
    \caption{The anti-Wick mapping taking a function $f(\alpha^\ast , \alpha)$ to an operator $\hat f = \mathscr{A} (f)$ based on $\alpha \to \hat A, \alpha^\ast \to \hat A^\ast$. We obtain $\mathscr{A}$ as a composition of an embedding $\varphi$ of commutative functions in a dilation followed by a partial vacuum expectation $\mathscr{E}$.}
    \label{fig:enter-label}
\end{figure}

The goal of this paper is to extend the technique to a general setting.
The outline of this paper is as follows. After establishing notation, we recall in Section \ref{sec:Complex_Wave} the complex-wave representation for Bose systems. We introduce the Hilbert space $L^2 (\mathbb{C}, \mathbb{P})$ of functions on the complex plane that are square-integrable with respect to a Gaussian measure $\mathbb{P}$. In particular, the operators $\hat A$ and $\hat B$ are $\frac{\partial}{\partial \alpha^\ast}$ and $\frac{\partial}{\partial \alpha}$, respectively, and we show that $\hat C$ corresponds to multiplication by $\alpha$. We give a reproducing kernel Hilbert space construction for the kernel $K(\alpha^\ast, \beta)= e^{\alpha^\ast \beta}$ where the reproducing kernel Hilbert space is the anti-holomorphic functions in $L^2 (\mathbb{C}, \mathbb{P})$. This leads to the Bargmann-Segal isomorphism with Fock space and we indicate how to extend this from functions on the complex plane to functions over a separable Hilbert space. 
In Section \ref{sec:anti-Wick}, we give the general formulation of anti-Wick ordering for bosonic fields over Fock space.

\bigskip

\textbf{Notation}
We shall work with Hilbert spaces over real and complex spaces with Gaussian measures. To begin with, we consider $L^2 (\mathbb{R} , \gamma )$ where $\gamma (dx) = \frac{e^{-x^2/2}}{\sqrt{2 \pi}} dx$ is the standardized Gaussian measure. Here we may defined annihilation and creation operators
\begin{eqnarray}
    \hat a \equiv \frac{\partial }{\partial x}, \quad
    \hat a^\ast \equiv x - \frac{\partial }{\partial x}
    \label{eq:aandadag}
\end{eqnarray}
and these satisfy the canonical commutation relations $[\hat a , \hat a^\ast]= \hat I$. The vacuum vector is given by $\Omega (x) = 1$ and we have $ \hat a \, \Omega =0$. We may set $\hat a = \frac{1}{\sqrt{2}} (\hat q +i \hat p )$ where we have the self-adjoint operators 
\begin{eqnarray}
    \hat q = \frac{1}{\sqrt{2}} (\hat a + \hat a^\ast ) \equiv \frac{1}{\sqrt{2}} x,
    \qquad
    \hat p =
\frac{1}{i} \frac{1}{\sqrt{2}}  ( \hat a - \hat a^\ast )\equiv -i ( \sqrt{2} \frac{\partial }{\partial x} -\frac{1}{\sqrt{2}} x)
\label{eq:qandp}
\end{eqnarray}
They satisfy $[\hat q, \hat p] =i \hat I$.

The number operator $\hat N = \hat a^\ast \hat a$ has a complete set of orthonormal eigenvectors $ |e_n \rangle$ : $(\hat N  - n ) \, | e_n \rangle =0$ for $n=0,1,2, \cdots$. We also note that the vector $ e_0 \equiv 1$ is the vacuum and is annihilated by $\hat a$; also that $\hat a \, |e_n   \rangle= \sqrt{n} \, | e_{n-1}  \rangle$ for $n \ge 1$.

\bigskip

We shall consider the Hilbert space over the complex plane. In the following, we shall take $\alpha = \frac{1}{\sqrt{2}}(x+iy)$ to be a complex number with $x,y$ (referred to as the quadratures) being real. We consider functions on the complex plane $\mathbb{C}$ of the form $f(\alpha^\ast , \alpha) = \sum_{j,k\ge0} f_{jk} (\alpha ^\ast )^j \alpha^k$. We may treat the variables $\alpha$ and $\alpha^\ast$ as independent variables if desired, and here we mean $f(\beta^\ast , \alpha) = \sum_{j,k\ge0} f_{jk} (\beta ^\ast )^j \alpha^k$.
It is also understood as a function $F=F(x,y)$ of the two quadratures, however, we typically write it as $f=f(\alpha^\ast , \alpha)\equiv F\big( \frac{1}{\sqrt{2}}(\alpha + \alpha^\ast) ,\frac{1}{i\sqrt{2}}(\alpha - \alpha^\ast )\big)$. If the dependence is only on $\frac{1}{\sqrt{2}}(x+iy)$ then we say $f$ is holomorphic and write $f=f( \alpha )$; likewise, if the dependence is only on $\frac{1}{\sqrt{2}}(x-iy)$ we say it is anti-holomorphic and write $f=f(\alpha^\ast )$. General functions can be written as 

\section{The Complex Wave Representation}
\label{sec:Complex_Wave}
A measure $\mathbb{P}$ on the complex plane is given by
\begin{eqnarray}
    \mathbb{P} [ d\alpha] = \mathbb{P} [d \alpha^\ast]
    = e^{-|\alpha |^2}\frac{d\alpha^2 }{\pi}
    \triangleq e^{-(x^2+y^2)/2}\frac{dxdy}{2 \pi}.
\end{eqnarray}

\begin{definition}
    The Hilbert space $L^2( \mathbb{C}, \mathbb{P})$ is defined to be the set of functions $f=f(\alpha , \alpha^\ast )$ satisfying
    \begin{eqnarray}
        \int_{\mathbb{C}} |f (\alpha , \alpha^\ast ) |^2 \, \mathbb{P} [ d \alpha] < \infty.
    \end{eqnarray}
     The inner product is then given by
     \begin{eqnarray}
         \langle f|g \rangle = \int_{\mathbb{C}} 
         f(\alpha , \alpha^\ast )^\ast g(\alpha , \alpha^\ast ) \,
         \mathbb{P} [ d \alpha ] \equiv  
         \int_{ \mathbb{R} \times \mathbb{R}}
         F(x,y)^\ast G(x,y) \,  \frac{e^{-x^2/2}dx}
         {\sqrt{2\pi}} \frac{e^{-y^2/2}dy}
         {\sqrt{2\pi}}.
     \end{eqnarray}
     The subspaces of holomorphic and anti-holomorphic functions are denoted as $L^2_{\mathrm{hol.}}( \mathbb{C}, \mathbb{P})$ and $L^2_{\mathrm{anti-hol.}}( \mathbb{C}, \mathbb{P})$, respectively.
\end{definition}

We note that constant functions form a one-dimensional subspace, and that this is the intersection of $L^2_{\text{hol.}}( \mathbb{C}, \mathbb{P})$ and $L^2_{\text{anti-hol.}}( \mathbb{C}, \mathbb{P})$.

\begin{definition}
    We define the following operators on $L^2( \mathbb{C}, \mathbb{P})$:
    \begin{eqnarray}
        \hat A = \frac{\partial}{\partial \alpha^\ast}, \qquad \hat A^\ast = \alpha^\ast - \frac{\partial}{\partial \alpha} , \qquad
        \hat B = \frac{\partial}{\partial \alpha}, \qquad \hat B^\ast = \alpha - \frac{\partial}{\partial \alpha^\ast } .
    \end{eqnarray}
\end{definition}

\begin{proposition}
    The pairs $(\hat A, \hat A^\ast)$ and $(\hat B , \hat B^\ast )$ are mutually adjoint and satisfy
    \begin{gather}
        [\hat A , \hat A^\ast] = \hat I = [ \hat B , \hat B^\ast ] ,\\
        [\hat A , \hat B] = [ \hat A , \hat B^\ast]= 
        [\hat A^\ast , \hat B] = [ \hat A^\ast , \hat B]=0.
    \end{gather}
    \end{proposition}
    
\begin{proof}
    In terms of the quadratures, we have $\frac{\partial}{\partial \alpha^\ast} = \frac{1}{\sqrt{2}} \big( \frac{\partial}{\partial x}+ i\frac{\partial}{\partial y}\big) $ and we compute
    \begin{eqnarray*}
        \langle \hat A f | g \rangle &=&
        \int_{\mathbb{R}^2}\bigg( \frac{1}{\sqrt{2}} \frac{\partial F^\ast}{\partial x}- i \frac{1}{\sqrt{2}} \frac{\partial F^\ast}{\partial y}\bigg)  G(x,y) \,
        e^{-(x^2 +y^2 )/2 }\frac{dxdy}{2 \pi} \nonumber \\
        &=&
        \int_{\mathbb{R}^2} F^\ast 
        \bigg\{ - \frac{1}{\sqrt{2}} \bigg( \frac{\partial G}{\partial x}- i  \frac{\partial G}{\partial y}\bigg) 
        + \frac{1}{\sqrt{2}} (x-iy) \, G(x,y) \bigg\}
        \,
        e^{-(x^2 +y^2 )/2 }\frac{dxdy}{2 \pi} 
        =
        \int_{\mathbb{C}} f^\ast \bigg( -\frac{\partial}{\partial \alpha} + \alpha^\ast \bigg) g \, \mathbb{P} [d \alpha ]
        .
    \end{eqnarray*}
    The case $\frac{\partial}{\partial \alpha} = \frac{1}{\sqrt{2}} \big( \frac{\partial}{\partial x}- i\frac{\partial}{\partial y}\big) $ is analogous. The commutation relations are easily established.
\end{proof}

It follows that we have two independent (commuting!) boson creation and annihilation operators on the Hilbert space $L^2( \mathbb{C}, \mathbb{P})$. Moreover, the holomorphic and anti-holomorphic spaces arise naturally as the subspaces annihilated:
\begin{eqnarray}
    L^2_{\text{hol.}}( \mathbb{C}, \mathbb{P}) &\equiv& \{ f \in  L^2( \mathbb{C}, \mathbb{P}) : \hat A f=0 \} , \\
    L^2_{\text{anti-hol.}}( \mathbb{C}, \mathbb{P}) &\equiv& \{ f \in  L^2( \mathbb{C}, \mathbb{P}) : \hat B f=0 \} .
\end{eqnarray}

\begin{proposition}
    The orthogonal projection $\mathcal{P}$ from
    $L^2( \mathbb{C}, \mathbb{P})$ onto $L^2_{\text{anti-hol.}}( \mathbb{C}, \mathbb{P})$ is extended by linearity from
    \begin{eqnarray}
        \mathcal{P} : (\alpha ^\ast)^n \alpha^m =
        \left\{
        \begin{array}{cc}
             \frac{n!}{(n-m)! } (\alpha^\ast )^{n-m}, &  n \ge m;\\
             0 , &  \text{otherwise}.
        \end{array}
        \right.
    \end{eqnarray}
\end{proposition}
\begin{proof}
    The functions $\frac{1}{\sqrt{k!}} (\alpha^\ast)^k$, $(k = 0,1,2,\cdots )$, form a complete orthonormal basis for $L^2_{\text{anti-hol.}}( \mathbb{C}, \mathbb{P})$ and, transferring to polar form $\alpha = re^{i \theta}$, we note that
    \begin{eqnarray*}
        \langle (\alpha^\ast)^k | (\alpha^\ast )^n \alpha^m \rangle
        = \int_{\mathbb{C}} (\alpha^\ast )^n \alpha^{m+k} \, \mathbb{P}[d \alpha ] = \int_0^\infty r dr \, e^{-r^2}r^{n+m+k}\int_0^{2\pi}\frac{d\theta}{2\pi} e^{i(m+k-n) \theta} .
    \end{eqnarray*}
The $\theta$ integral leads to $\delta_{m+k,n}$ while we note that $\int_0^\infty dr \, e^{-r^2}r^{2p+1}= p!$ for positive integers $p$. Therefore, $ \langle (\alpha^\ast)^k | (\alpha^\ast )^n \alpha^m \rangle  =\delta_{m+k,n} n!$.
We then have $\mathcal{P} = \sum_k \frac{1}{k!}| (\alpha^\ast)^k \rangle \langle (\alpha^\ast)^k| $ and the result follows.
\end{proof}

\bigskip

We now come to the main observation.

\begin{proposition}
    The operation of multiplication by $\alpha$ on $L^2( \mathbb{C}, \mathbb{P})$ is given by $\hat C = \hat A+ \hat B^\ast$ while multiplication by $\alpha^\ast $ is given by $\hat C ^\ast = \hat A ^\ast + \hat B$. The operators $\hat C$ and $\hat C^\ast$ are normal.
\end{proposition}
This is very easily seen:
\begin{eqnarray}
    [ \hat C , \hat C^\ast ] = [\hat A , \hat A^\ast ]+ [\hat B^\ast , \hat B] = \hat I - \hat I = 0.
\end{eqnarray}

In one sense this is fairly natural. We have $[\hat A , \hat B]=0$ which just says that $ \frac{\partial }{\partial \alpha^\ast}$ and $ \frac{\partial }{\partial \alpha}$ commute. Likewise $[\hat C , \hat C^\ast]=0$ is just saying that multiplication by $\alpha$ and multiplication by $\alpha^\ast$ also commute.

\subsection{Isomorphism}
We have the fairly natural isomorphism $L^2( \mathbb{C}, \mathbb{P}) \cong L^2 (\mathbb{R}\oplus i \, \mathbb{R} , \gamma \times \gamma ) \cong L^2 (\mathbb{R}, \gamma) \otimes L^2 (\mathbb{R} , \gamma) $ and it is instructive to represent the operators above in these terms.

Let us write $\hat a_x, \hat a_x^\ast, \hat q_x, \hat p_x$ for the analogues of the standard canonical operators (\ref{eq:aandadag}),(\ref{eq:qandp}) on the first factor, with $\hat a_y, \hat a_y^\ast, \hat q_y, \hat p_y$ the corresponding operators on the second factor.

We then have
\begin{eqnarray}
    \hat A = \frac{1}{\sqrt{2}} (\hat a_x \otimes \hat I +i \hat I \otimes \hat a_y ) , \quad
    \hat A^\ast = \frac{1}{\sqrt{2}} (\hat a_x ^\ast\otimes \hat I -i \hat I \otimes \hat a_y ^\ast) , \nonumber \\
     \hat B = \frac{1}{\sqrt{2}} (\hat a_x \otimes \hat I -i \hat I \otimes \hat a_y ) , \quad
    \hat B^\ast = \frac{1}{\sqrt{2}} (\hat a_x ^\ast\otimes \hat I +i \hat I \otimes \hat a_y^\ast )
\end{eqnarray}
We may split in quadratures: $\hat X_A = \frac{1}{\sqrt{2}} (\hat A + \hat A^\ast ) , \, \hat Y_A = \frac{1}{i} \frac{1}{\sqrt{2}} (\hat A - \hat A^\ast ) $, etc., in which case
\begin{eqnarray}
    \hat X_A = \frac{1}{\sqrt{2}} (\hat q_x \otimes \hat I - \hat I \otimes \hat p_y ) , \quad
    \hat Y_A = \frac{1}{\sqrt{2}} (\hat p_x \otimes \hat I + \hat I \otimes \hat q_y ) , \nonumber \\
     \hat X_B = \frac{1}{\sqrt{2}} (\hat q_x \otimes \hat I + \hat I \otimes \hat p_y ) , \quad
    \hat Y_B = \frac{1}{\sqrt{2}} (\hat p_x \otimes \hat I - \hat I \otimes \hat q_y ) .
\end{eqnarray}
The two normal operators are given by
\begin{eqnarray}
    \hat C=  \hat q_x \otimes \hat I +i \hat I \otimes \hat q_y  , \quad
    \hat Y_A =  
    \hat C^\ast =  \hat q_x \otimes \hat I -i \hat I \otimes \hat q_y .
\end{eqnarray}

\subsection{Reproducing Kernel Hilbert Spaces}
Let $\mathcal{X}$ be a set and $K : \mathcal{X} \times \mathcal{X}\mapsto \mathbb{C}$ be a (positive definite) kernel, that is, 
\begin{eqnarray}
    \sum_{j,k}c_j^\ast K (x_j , x_k ) c_k \ge 0,
\end{eqnarray}
for each integer $n \ge 1$ and $x_1, \cdots , x_n \in \mathcal{X}$ and $c_1, \cdots , c_n \in \mathbb{C}$. A (not necessarily separable) Hilbert space of functions $\mathfrak{K}$ on $\mathcal{X}$ is called a reproducing kernel Hilbert space (RKHS) for $K$ if it includes the representer functions $\Bbbk_x (\cdot ) = K (\cdot , x)$ and $\langle \Bbbk_x | f \rangle = f (x)$ for all $f \in \mathfrak{K}, x \in \mathcal{X}$. In particular, $\langle \Bbbk_x | \Bbbk_y \rangle = K(x,y)$.

We note that $K(x,y)^\ast = K(y,x)$ with the complex conjugate kernel $K^\ast$ again being a kernel.
Its representers will be $\widetilde{\Bbbk}_x (y) \equiv \Bbbk_y (x)$.

\bigskip

Our starting point is the fact that a kernel on $\mathbb{C}$ is given by
\begin{eqnarray}
    K(\alpha ^\ast, \beta ) = e^{\alpha^\ast \beta}.
\end{eqnarray}
Note that we depart from the standard RKHS notation to emphasize that, in this specific case, the kernel is \textit{anti}-holomorphic function in its first argument. This is in line with our earlier conventions. The corresponding representer is likewise denoted as $\Bbbk_\beta (\alpha^\ast ) = e^{\alpha^\ast \beta }$. The RKHS $\mathfrak{K}$ will be spanned by the representers  and must therefore be a Hilbert space of anti-holomorphic functions. A routine calculation shows that $\int_{\mathbb{C}} e^{\beta^\ast \alpha }e^{\alpha^\ast \gamma} \, \mathbb{P} [ d \alpha ] = e^{\beta^\ast \gamma}$, or
\begin{eqnarray}
    \int_{\mathbb{C}} K(\alpha^\ast , \gamma ) K (\gamma^\ast , \beta ) \, \mathbb{P} [ d \gamma ]
    =K (\alpha ^\ast , \beta ) ,
\label{eq:K_pop}
\end{eqnarray}
and from this we see that the RKHS here is, in fact, $\mathfrak{K}= L^2_{\mathrm{anti-hol.}}( \mathbb{C}, \mathbb{P})$. Indeed, (\ref{eq:K_pop})
gives us $\langle \Bbbk_\alpha  | \Bbbk_\beta \rangle_{\mathfrak{K}} = K( \alpha^\ast  , \beta )$
and that
\begin{eqnarray}
    \langle \Bbbk_\alpha | f \rangle_{\mathfrak{K}} = f ( \alpha^\ast ) , \qquad \big( f \in {\mathfrak{K}} \big) .
\end{eqnarray}

To see how this comes about, let us recall that the exponential vectors are defined by $| \exp (\alpha ) \rangle = e^{\hat a^\ast \alpha} |e_0  \rangle= \sum_n \frac{\alpha^n}{\sqrt{n!}} \,| e_n  \rangle$ and we have $\big( \hat a - \alpha \big) \, | \exp (\alpha ) \rangle =0$. 
One easily sees that $\langle \exp (\alpha ) | \exp (\beta ) \rangle = e^{\alpha^\ast \beta}$.

We define a unitary map $U : \mathcal{H} \mapsto \mathfrak{K} =L^2_{\mathrm{anti-hol.}}( \mathbb{C}, \mathbb{P}) $ by $U: | \psi \rangle \mapsto f_\psi = f_\psi(\alpha ^ \ast )$ where
\begin{eqnarray}
    f_\psi( \alpha ^\ast ) = \langle \exp (\alpha ) | \psi \rangle_{\mathcal{H}}= \sum_n \langle n | \psi \rangle_{\mathcal{H}}  \frac{( \alpha^\ast)^n}{\sqrt{n!}} .
\end{eqnarray}
Unitary follows from the fact that the anti-holomorphic functions $\frac{( \alpha^\ast)^n}{\sqrt{n!}}$, ($n =0, 1,2, \cdots $) form a complete orthonormal basis for $L^2_{\mathrm{anti-hol.}}( \mathbb{C}, \mathbb{P}) $.

The representers are then $ \Bbbk_\beta = U\, | \exp (\beta ) \rangle $ so $  \Bbbk_\beta (\alpha^\ast ) = \sum_n \frac{\beta^n}{\sqrt{n!}} \frac{( \alpha^\ast)^n}{\sqrt{n!}} = e^{\alpha^\ast \beta}$, as required.    
One then sees that $\langle \Bbbk_\beta | f_\psi \rangle_{\mathfrak{K} } = \langle \exp (\beta ) | \psi \rangle_{\mathcal{H}} $ which is $f_\psi (\beta^\ast )$ by definition, and
$  \langle \Bbbk_\beta | \Bbbk_\gamma  \rangle_{\mathfrak{K} } = \langle \exp (\beta ) | \exp (\gamma )  \rangle_{\mathcal{H}} = e^{\beta^\ast \gamma}\equiv K( \beta , \gamma ) $.

The normalized versions of the exponential vectors are known as coherent states. The above construction is known as the Bargmann-Segal representation of $\mathcal{H}$.

The creator $\hat a^\ast$ is represented by $U\hat a^\ast U^{-1}$ and this corresponds to multiplication by $\alpha$, or more broadly the restriction of $\hat A^\ast= \alpha^\ast - \frac{\partial }{\partial \alpha}$ to the anti-holomorphic functions. This follows from the fact that $\langle \exp (\alpha ) | \hat a ^\ast|\psi \rangle= \alpha^\ast \, \langle \exp (\alpha ) |\psi \rangle$. From the fact that $ \hat a^\ast \, | \exp (\alpha ) \rangle \equiv \frac{\partial }{\partial \alpha } | \exp (\alpha ) \rangle$, we find that $U \hat a\, U^{-1}$ corresponds to the operation $\frac{\partial }{\partial \alpha^\ast }$.

\bigskip

We likewise consider the complex conjugate kernel $\widetilde{K} (\alpha , \beta ) = K (\beta , \alpha  )$, and this time the RKHS $\widetilde{\mathfrak{K}}$ is given by the holomorphic functions $L^2_{\mathrm{hol.}} ( \mathbb{C}, \mathbb{P})$.
Here the unitary map $\widetilde{U} : \mathcal{H} \mapsto L^2_{\mathrm{hol.}} ( \mathbb{C}, \mathbb{P}) $ is $\widetilde{U}: | \psi \rangle \mapsto g_\psi = g_\psi (\alpha )$ where $ g_\psi( \alpha  ) = \langle \exp (\alpha^\ast ) | \psi \rangle_{\mathcal{H}}$.
The representers are now $\widetilde{\Bbbk}_\beta = \widetilde{U} \, | \exp (\beta ^\ast ) \rangle $.

\subsection{Generalization to Functions over a Hilbert Space}
It is possible to generalize the Bargmann-Segal construction from Hilbert spaces over the complex numbers $\mathbb{C}$ to those over a separable Hilbert space $\mathfrak{h}$. 

A kernel on $\mathfrak{h}$ is then given by $K(\phi , \psi ) =e^{\langle \phi , \psi \rangle}$. The RKHS will be a straightforward tensor product $\bigotimes^n L^2 (\mathbb{C}, \mathbb{P})$ where $\mathfrak{h} \cong \mathbb{C}^n$, however, we have to deal with a non-separable Hilbert space when $\mathfrak{h}$ is infinite dimensional. The relation $\int_{\mathfrak{h}} K(\phi, \psi ) K(\psi , \phi') \, \mathbb{P}_{\mathfrak{h}} [ d \psi ]= K(\phi , \phi')$ generally only defines a pre-measure in this case $\mathbb{P}_{\mathfrak{h}} $. It is possible to construct a $\sigma$-additive measure over a larger Hilbert space $\mathfrak{h}^>$, see for instance \cite{GK2018}.

In place of the single mode Hilbert space $\mathcal{H}$ we take the (Bose) Fock space $\Gamma (\mathfrak{h}) = \bigoplus_{n=1}^\infty ( \otimes_{\text{symm.}}^n \mathfrak{h})$ with $\mathfrak{h}$ as one-particle space and define the exponential vectors with test function $\phi \in \mathfrak{h}$ by
\begin{eqnarray}
    | \exp (\phi ) \rangle = 1 \oplus \phi \oplus \frac{(\phi \otimes \phi)}{\sqrt{2!}}\oplus \frac{(\phi \otimes\phi \otimes \phi)}{\sqrt{3!}} \cdots .
\end{eqnarray}
It follows that $\langle \exp (\phi ) | \exp (\psi )\rangle = K (\phi , \psi )$.
\section{Anti-Wick Quantization}
\label{sec:anti-Wick}
Let us return to the single mode case ($\mathfrak{h}\cong \mathbb{C}$). It is convenient to introduce the unimodular function
\begin{eqnarray}
    \varpi (\alpha , \beta ) = \frac{K(\alpha^\ast , \beta)}{K(\beta^\ast , \alpha )}
    =\frac{\Bbbk_\beta (\alpha^\ast)}{\Bbbk_\alpha (\beta^\ast)}
    = e^{\alpha^\ast \beta - \beta^\ast \alpha}.
\end{eqnarray}
This allows us to introduce a Fourier transform pair on $L^2 ( \mathbb{C}, \mathbb{P})$ as
\begin{eqnarray}
    \widetilde{f} ( \alpha^\ast , \alpha )
    &=&\int_{\mathbb{C}} f(\beta^\ast , \beta ) \, \varpi (\alpha , \beta ) \, \mathbb{P} [ d \beta ] \nonumber \\
    f ( \alpha^\ast , \alpha )
    &=&\int_{\mathbb{C}} \widetilde{f}(\beta^\ast , \beta ) \, \varpi (\alpha , \beta ) \, \mathbb{P} [ d \beta ].
\end{eqnarray}

Taking $\mathcal{H}$ to again be the Hilbert space on which creation and annihilation operators $\hat a^\ast ,\hat a$ act, we may define a quantization rule $\mathscr{Q}$ replacing $\alpha , \alpha^\ast$ with $\hat a , \hat a^\ast$ as follows:
\begin{eqnarray}
    \mathscr{Q} (f) \triangleq
    \int_{\mathbb{C}} \widetilde{f}(\beta^\ast , \beta ) \, \mathscr{Q} \big( \varpi (\cdot , \beta ) \big) \, \mathbb{P} [ d \beta ]
\end{eqnarray}
where we must give the value of $\mathscr{Q} \big( \varpi (\cdot , \beta ) \big)$ for each $\beta \in \mathbb{C}$. In particular, we have the following choices
\begin{eqnarray}
    \mathscr{W} \big( \varpi (\cdot , \beta ) \big)
    &=&  e^{ \hat a^\ast \beta - \beta^\ast \hat a} , \quad (\mathrm{Weyl}) \nonumber \\
    \mathscr{N} \big( \varpi (\cdot , \beta ) \big)
    &=&  e^{ \hat a^\ast \beta } e^{- \beta^\ast \hat a}, \quad (\mathrm{Wick}) \nonumber \\
    \mathscr{A} \big( \varpi (\cdot , \beta ) \big)
    &=&  e^{  - \beta^\ast \hat a} e^{\hat a^\ast \beta } , \quad (\mathrm{Anti-Wick}) .
\end{eqnarray}

It is easy to see that, for $f(\alpha^\ast , \alpha) = \sum_{j,k\ge0} f_{jk} (\alpha ^\ast )^j \alpha^k$, we obtain $\mathscr{N} (f) = \sum_{j,k\ge0} f_{jk} (\hat a ^\ast )^j \hat a^k$.
For anti-Wick, we have 
\begin{eqnarray}
    \frac{\langle \exp (\alpha ) | \mathscr{A}(f) | \exp (\beta ) \rangle }{\langle \exp (\alpha ) | \exp (\beta ) \rangle} = f (\alpha^\ast , \beta ).
    \label{eq:exp_expect}
\end{eqnarray}

We recall the well-known integral form for the anti-Wick rule for the single oscillator mode (see Theorem 3, chapter 3 of Louisell\cite{Louisell}, or de Gosson \cite{de_Gosson})
\begin{eqnarray}
\mathscr{A}(f) = \int_{\mathbb{C}} f(\alpha^\ast , \alpha ) | \text{exp} (\alpha )\rangle \langle \text{exp} (\alpha )| \, \mathbb{P} [d\alpha ].
\label{eq:anti_Wick_exp}
\end{eqnarray}
Note that $\langle \psi | \mathscr{A}(f) \, \psi \rangle = \int_{\mathbb{C}} f(\alpha^\ast , \alpha ) \, \langle \psi | \text{exp} (\alpha )\rangle |^2 \, \mathbb{P} [d\alpha ]$ so we have the positivity condition $\mathscr{A}(f) \ge 0$ whenever $f \ge 0$.

\bigskip

The relation (\ref{eq:exp_expect}) generalizes to Fock space over infinite-dimensional Hilbert spaces and gives us the de-quantization scheme. However, generalizing (\ref{eq:anti_Wick_exp}) to the infinite dimensional case is challenging and we formalize the anti-Wick quantization scheme in this setting in the next section.

\subsection{Anti-Wick Ordering of Fields}
In the following, we fix a separable Hilbert space $\mathfrak{h}$ with an anti-linear involutive map $J$, so that $\langle J \phi | J \psi \rangle =\langle \psi | \phi \rangle$. 

We consider two copies of the (bosonic) Fock space over $\mathfrak{h}$ which we denote by $\mathfrak{F}_{A}$ and $\mathfrak{F}_{B}$, respectively. The annihilation and creation operators on $\mathfrak{F}_{A}$ are denoted as $A(\cdot )$ and $A^\ast (\cdot )$, and we have the commutation relations $[A(\phi ), A^\ast (\phi ) ] = \langle \phi | \psi \rangle$. The corresponding operators on $\mathfrak{F}_{B}$ are denoted by $B(\cdot )$ and $B^\ast (\cdot )$.

On the tensor product $\mathfrak{F}_{A}\otimes \mathfrak{F}_{B}$ we consider the operators 
\begin{eqnarray}
Z\left( \phi \right) =A\left( \phi \right) \otimes I_{B}+I_{A}\otimes B\left( J \phi \right) ^{\ast }.
\end{eqnarray}
We remark that $Z (\cdot )$ is anti-linear: the creator fields being linear and the annihilation fields anti-linear in their arguments.

\begin{proposition}
The von Neumann algebra $\mathfrak{Z}$ generated by the operators $Z ( \cdot ) $ is commutative.

\begin{proof}
We have $\mathfrak{Z} = \left\{ Z(\phi): \phi \in \mathfrak{h}  \right\} ^{\prime \prime
} $ by the von Neumann bi-commutant Theorem, and we note that this must be a
*-algebra. However, we note the identity 
\begin{eqnarray}
\left[ Z\left( \phi\right) ,Z\left( \psi\right) ^{\ast }\right] =[A\left( \phi\right)
,A\left( \psi \right) ^{\ast }]\otimes I_{B}+I_{A}\otimes [ B( J\phi ),B\left(
J \psi \right) ^{\ast }]= \langle \phi | \psi \rangle \, I_{A}\otimes I_{B}-I_{A}\otimes
\langle J \psi | J \phi \rangle \, I_{B}=0.
\end{eqnarray}
This ensures that the fields $Z (\cdot )$ always
commutes with the adjoint fields $Z^{\ast }(\cdot )$ for all arguments.
\end{proof}
\end{proposition}

\bigskip

Let us recall that we may introduce quadrature process $Q_{A}\left( \phi\right)
=A\left( \phi \right) +A^{\ast }\left( \phi \right) $ and $P_{A}\left( \phi \right) =%
\frac{1}{i}\left( A\left( \phi \right) -A\left( \phi \right) ^{\ast }\right) $ on
Fock space $\mathfrak{F}_{A}$. The field $Q_{A}(\cdot )$ is self-commuting and,
for the Fock vacuum state $\Omega _{A}$ has the statistics of a Gaussian field:
\begin{eqnarray}
    \langle \Omega_A |e^{iu \, Q_A (\phi )} \, \Omega_A \rangle =
    \langle \Omega_A |e^{iu \, P_A (\phi )} \, \Omega_A \rangle =e^{- u^2 \| \phi \| ^2 /2}.
\end{eqnarray}
process. The fields $Q_A$ and $P_{A}$, however, are non-commuting: 
\begin{eqnarray}
\left[ Q_{A}(\phi ),P_{A}\left( \psi \right) \right] =2i \,\text{Re} \langle \phi | \psi \rangle \, I_{A}.
\end{eqnarray}
The process $Z$ has quadratures 
\begin{eqnarray*}
Q_{Z}(\phi ) &=& Z (\phi )+Z^{\ast }(\phi ) =Q_{A}( \phi )\otimes I_{B}+I_{A}\otimes Q_{B} (J \phi ), \\
P_{Z}(\phi ) &=&\frac{1}{i}\left( Z(\phi )-Z^{\ast }(\phi )\right) =P_{A}(\phi ) \otimes
I_{B}+I_{A}\otimes P_{B}(J \phi ),
\end{eqnarray*}
and by inspection these quadratures do in fact commute: $\left[ Q_{Z}(\phi ) ,P_{Z}(\psi )\right] \equiv 0$.

We now introduce the mapping $\mathcal{E} _{A}$ from $\mathfrak{Z}$ to the von Neumann algebra $
\mathfrak{A}$ generated by the $A( \cdot )$ defined by 
\begin{eqnarray}
\langle u|\,\mathcal{E} _{A}(X)\,|v\rangle _{B}=\langle u \otimes \Omega _{B} |\,X\,|v \otimes \Omega _{A}\rangle ,
\end{eqnarray}
for arbitrary $u,v\in \mathfrak{F}_{B}$.

\begin{remark}
The mapping $\mathcal{E}_A$ is a projection and may alternatively be written in terms of a partial trace as follows $\mathrm{tr}_B \big(I_A \otimes  |\Omega_B \rangle \langle  \Omega_B | \, X \big)$.
We may similarly define the mapping $\mathcal{E}_B$ onto the corresponding algebra $\mathfrak{B}$.
\end{remark}

\begin{proposition}
The restriction of the mapping $\mathcal{E} _{A}$ to $\mathfrak{Z} $ replaces a
function of the commuting fields $Z(\cdot )$ and $Z^{\ast }(\cdot)$ with the
corresponding anti-Wick ordered version of the function in terms of the fields $A(\cdot )$ and $
A^{\ast } (\cdot )$.
\end{proposition}

\begin{proof}
It suffices to establish this for monomials: set $X=Z\left( \phi_{1}\right) \cdots Z\left( \phi_{n}\right)
Z\left( \psi_{1}\right) ^{\ast }\cdots Z\left( \psi_{m}\right) ^{\ast }$ which is
affiliated to $\mathfrak{Z}$. (Note that we
may reorder the terms in the monomial since the terms all commute!) Let us
do this so that the $B$ fields are all in Wick order, viz., 
\begin{eqnarray}
 \prod_{j}\left( A\left( \phi_{j}\right) \otimes
I_{B}+I_{A}\otimes B\left( J\phi_{j}\right) ^{\ast }\right) \prod_{k}\left( A\left( \psi_{k}\right) ^{\ast }\otimes I_{B}+I_{A}\otimes
B\left( J\psi_{k}\right) \right).
\end{eqnarray}
Taking the Fock vacuum expectation with respect to $\mathfrak{F}_{B}$ then
results in all nontrivial Wick-ordered terms in the $B$ processes vanishing identically
leaving 
\begin{eqnarray}
\mathcal{E} _{A}\left( X\right) =\prod_{j}A\left( \phi_{j}\right) \prod_{k}A\left(
\psi_{k}\right) ^{\ast } .
\end{eqnarray}
\end{proof}

Similarly, $\mathcal{E} _{B}\left( X\right) =\prod_{k}B\left( J\psi_{k}\right) \prod_{j}B\left(
J\phi_{j}\right) ^{\ast } $.

\begin{remark}
The restrictions $\mathcal{E}_A$ and $\mathcal{E}_B$ are not projections since $\mathfrak{A} $ and $\mathfrak{B}$ are not subspaces of $\mathfrak{Z}$. Indeed, the restrictions take us from a commutative algebra to a non-commutative one! Likewise, they are not conditional expectations on von Neumann algebras.
\end{remark}

\begin{corollary}
    The anti-Wick quantization is a completely positive map.
\end{corollary}
\begin{proof}
    We have that $\varphi$ is a positive isomorphism from a commutative algebra of bounded functions to a commutative sub-algebra of $\mathfrak{A}\otimes \mathfrak{B}$, and is therefore automatically completely. The partial vacuum expectation $\mathcal{E}_A$ from $\mathfrak{A}\otimes \mathfrak{B}$ onto $\mathfrak{A}$ is also completely positive. Therefore, the composition $\mathscr{A}=\mathcal{E}\circ \varphi$ must also be completely positive.
\end{proof}

\section{Discussion and Conclusion}
Let us, for definiteness, take $\mathfrak{h}=L^2 (\mathcal{X},dx)$ where $\mathcal X$ is a measurable space and $dx$ a $\sigma$-finite measure. We consider fields $\alpha (\phi )= \int_{\mathcal{X}} \phi (x)^\ast \alpha _x \, dx$ and $\alpha (\phi )^\ast = \int_{\mathcal{X}} \phi (x) \alpha _x^\ast  \, dx$ where $\alpha_\cdot$ is an essentially bounded complex-valued function on $\mathcal{X}$. (Note that the fields $\alpha (\cdot )$ and $\alpha (\cdot )^\ast$ are commuting.) The Weyl quantization rule is extended from $e^{\alpha^\ast (\phi )- \alpha (\phi )}\mapsto W(\phi)= e^{A(\phi )^\ast - A(\phi )} = e^{-iP(\phi )}$. Here, $W(\cdot )$ gives the representation of the Weyl unitaries on $\Gamma (\mathfrak{h})$. A quantization rule $\mathscr{Q}$ is said to be Cohen class if it is extended from $\mathscr{Q}: e^{\alpha^\ast (\phi )- \alpha (\phi )}\mapsto \Xi (\phi ) \, W(\phi)$ where $\Xi$ is a complex-valued function called a Cohen multiplier.

The anti-Wick quantization extends from $\mathscr{A} :e^{\alpha^\ast (\phi )- \alpha (\phi )}\mapsto e^{-A (\phi )} e^{A(\phi )^\ast }$ and this corresponds to the Cohen multiplier
\begin{eqnarray}
    \Xi _{\text{a.-W.}}(\phi ) = e^{- \frac{1}{2} \langle \phi |\phi \rangle}.
\end{eqnarray}
An important functorial property is that if $\mathfrak{h}=\mathfrak{h}_1\otimes \mathfrak{h}_2$ then $\Gamma (\mathfrak{h} )\cong \Gamma (\mathfrak{h}_1) \otimes \Gamma (\mathfrak{h}_2)$. In the anti-Wick case, we note that the multipliers have the property
\begin{eqnarray}
    \Xi _{\text{a.-W.}}(\phi \oplus \phi_2 ) \equiv \Xi _{\text{a.-W.}}(\phi_1 ) \otimes  
    \Xi  _{\text{a.-W.}}(\phi_2 )  .
    \label{eq:Xi_split}
\end{eqnarray}
Taking $\mathfrak{h}_k=L^2 (\mathcal{X}_k,dx_k)$
and $\mathfrak{h}_1\otimes \mathfrak{h}_2 \cong L^2 (\mathcal{X}_1 \times \mathcal{X}_2,dx_1 \times dx_2)$, then (\ref{eq:Xi_split}) implies the following factor property for anti-Wick
\begin{eqnarray}
    \mathscr{A}( \phi_1 \, \phi_2 ) = \mathscr{A} (\phi_1 ) \otimes \mathscr{A} (\phi_2 )
    \label{eq:anti_Wick_factor}
\end{eqnarray}
where $\phi_1 \, \phi_2 : (x_1, x_2) \mapsto \phi_1(x_1)\phi_2 (x_2)  $.
Note that (\ref{eq:Xi_split}) is a specific property that not all multipliers have: for instance, the well-known Born-Jordan quantization rule does not satisfy this and ``entangles'' the two factors \cite{de_Gosson}.

The factor property (\ref{eq:anti_Wick_factor}) combined with positivity of the anti-Wick scheme sows that it is completely positive in the finite dimensional scheme. Our corollary establishes the infinite dimensional case.

\end{document}